\documentclass[review]{elsarticle}

\usepackage{hyperref}

\usepackage{amsmath}

\usepackage{color}
\usepackage{soul}
\newcommand\hlbreakable[1]{\textcolor{black}{#1}}

\usepackage[mathdef, noreview]{custom}
\usepackage{accents}
\newcommand{\ubar}[1]{\underaccent{\bar}{#1}}

\journal{Computers \& Security}

\biboptions{authoryear}

\DeclareRobustCommand\EIOPAlongname{}

\newcommand{\var}{\operatorname{Var}}

\begin{document}

\begin{frontmatter}

\title{Duopoly insurers' incentives for data quality under a mandatory cyber data sharing regime}


\author[KTH]{Carlos Barreto}
\ead{cbarreto@kth.se}

\author[UMU,LF]{Olof Reinert}

\author[UMU,LF]{Tobias Wiesinger}

\author[KTH,RISE]{Ulrik Franke}

\address[KTH]{Digital Futures, KTH Royal Institute of Technology, SE-100 44 Stockholm, Sweden}
\address[UMU]{Ume{\aa} University, SE-901 87 Ume{\aa}, Sweden}
\address[LF]{L{\"a}nsf{\"o}rs{\"a}kringar, SE-106 50 Stockholm, Sweden}
\address[RISE]{RISE Research Institutes of Sweden, P.O. Box 1263, SE-164 29 Kista, Sweden}

\begin{abstract}
We study the impact of data sharing policies on cyber insurance markets. These policies have been proposed to address the scarcity of data about cyber threats, which is essential to manage cyber risks. We propose a Cournot duopoly competition model in which two insurers choose the number of policies they offer (i.e., their production level) and also the resources they invest to ensure the quality of data regarding the cost of claims (i.e., the data quality of their production cost). We find that enacting mandatory data sharing sometimes creates situations in which at most one of the two insurers invests in data quality, whereas both insurers would invest when information sharing is not mandatory. This raises concerns about the merits of making data sharing mandatory.
\end{abstract}


\end{frontmatter}


\section{Introduction \label{intro}}

With an ever-increasing societal dependence upon information technology and digital services, \emph{cyber risk} has received much attention lately. Such risk is not limited to any particular sector, but can be found everywhere; from manufacturing \citep{wells2014cyber,Ani2017cybersecurity} over healthcare \citep{kruse2017cybersecurity,coventry2018cybersecurity} and the power grid \citep{ericsson2010cyber,sridhar2011cyber} to financial services \citep{Kopp2017,dupont2019cyber,varga2021cyber}.

Not only is cyber risk increasingly found everywhere, but the interconnectedness and interdependency of this modern world also poses challenges of its own. As pointed out by \cite{bohme2006models}, there are at least two important forms of \emph{interdependent cyber risk}: First, firms are connected to each other. While allowing huge efficiency gains when exchanging information across complex supply-chains, this also means that diligent security efforts at any one firm always risk being undermined by sloppy security somewhere else. The proverbial chain is never stronger than its weakest link. Second, many firms use the same systems, so a vulnerability found in a popular operating system, web browser, or encryption protocol may immediately put millions if not billions of machines at risk.

These difficulties are particularly relevant for insurers and reinsurers who underwrite cyber risks as part of \emph{cyber insurance} offerings, and it has been repeatedly observed that interdependent cyber risk poses an important challenge to the development of a more mature and well-functioning cyber insurance market \citep[pp.~93--98]{anderson2006economics, oecd2017cyber}. \hlbreakable{This is not the place to review the extensive literature on cyber insurance---a comprehensive but slightly dated literature review is offered by \citet{bohme2010modeling} and a more recent review is given by \citet{marotta2017cyber,barreto2021cyber}. Some notable complications with cyber insurance---in addition to the interdependence of cyber risks noted above---include unclear coverage, immature market offerings, various information asymmetries, and lack of cyber security experience and expertise on the part of insurers. In our context, however, the most important  complication is lack of good actuarial data \citep[see e.g.][pp.~94--95]{biener2015insurability,franke2017cyberinsurance,eiopa2019cyber,oecd2017cyber}.}

\hlbreakable{
To some extent, this lack of data reflects more general problems with cyber risk data, not limited to cyber insurance. A recent attempt to systematize quantitative studies of the consequences of cyber incidents by \citet{woods2021systematization} found several contradictory and sometimes spurious results, and cautions against employing too simple statistical relationships. Similarly, a review of estimates of cyber risk likelihood found contradicting trends and emphasizes the need for rigorous and transparent methods to avoid jumping to erroneous conclusions \citep{woods2022reviewing}. In the insurance context, the difficulty of properly quantifying cyber risk forces expert-based or best-guess rather than actuarial pricing. Clearly, this may lead to undesirable outcomes, such as underpricing, where insurers unknowingly accept too much cyber risk, overpricing, where insureds pay too much for their risk transfer, or blanket exclusions of certain kinds of customers, who thus cannot reap the benefits of insurance \citep[see, e.g.,][]{gordon2003framework,mott2023between}.}
Against this background, it has been proposed that increased sharing of data between insurers might be beneficial.

A recent example is an analysis by the \cite{oecd2020cyberdata} on how to enhance the availability of data for cyber insurance underwriting. The report walks through existing practices such as cyber incident data being published by CERTs or regulators, information exchange (such as the CRO forum), commercial  catastrophe  models made available by firms such as AIR Worldwide and RMS, and reinsurer collections of aggregate data, but ultimately concludes that ``[n]one of these data sources on their own provide sufficient information for underwriting coverage as incident data is seen to be incomplete, historical experience covers too few claims and models are relatively new and untested'' \cite[p.~9]{oecd2020cyberdata}. Instead, three recommendations for government action are made; (i)~to remove legal obstacles to incident and claims data sharing, (ii)~to encourage industry associations to establish mechanisms for incident and claims data sharing, and (iii)~to encourage international collaboration.

Another recent example is a strategy note on cyber underwriting published by the European Insurance and Occupational Pensions Authority \citep{eiopa2020cyber}. Here, lack of data is identified as a primary obstacle to the understanding of cyber risk, and accordingly, to appropriate coverage being offered on the market. It is also noted that the mandatory incident reporting regimes established by  recent legislation such as the GDPR and the NIS directive will create relevant data. Against this background it is argued that access to a cyber incident database ``could be seen as a public  good  and  underpin  the  further  development  of the European cyber insurance industry and act as an enabler of the digital economy'' \citep[p.~3]{eiopa2020cyber}. The strategy delineated consists of EIOPA (i)~promoting a harmonized cyber incident reporting taxonomy with ``an aim to promote the development of a centralised (anonymised) database'' \citep[p.~4]{eiopa2020cyber}, (ii)~engaging with the industry to understand their perspective, and (iii)~encouraging data sharing initiatives.

The industry association \cite{insuranceeurope2020cyber}, in a direct response to the EIOPA strategy, broadly welcomes the strategy's recognition that lack of data is a serious impediment to the growth of the European cyber insurance market. However, Insurance Europe also notes that there are trade-offs involved. Specifically, it is cautioned that while a common cyber incident database should ideally be more detailed than the GDPR and NIS data it should at the same time not impose unnecessary burdens of additional reporting or IT system adaptation, and such a database should not  distort  competition.\footnote{Despite the enthusiasm of Insurance Europe about using GDPR and NIS incident reports to improve cyber insurance offerings, not everyone offering cyber insurance is even aware of this possibility, as shown by a study in Norway where the interviewed ``insurers seem oblivious to this aspect of NIS'' \citep{bahsi2019cyberinsurance}.} Specifically, ``if  an  insurer shares data it must gain access to an equal quantity and quality of data in return'' \cite[p.~2]{insuranceeurope2020cyber}.

The relevance of data quality is underscored by recent empirical research on NIS incident reports. Based on all the mandatory NIS incident reports received by the responsible government agency in Sweden in 2020, \citet{Franke2021NIS} find the economic aspects of reports to be incomplete and sometimes difficult to interpret. Thus, it is concluded that ``just making NIS reporting, as-is, available to insurers would not by itself solve the problem of lack of data for cyber insurance. Making the most of the reporting requires additional quality assurance mechanisms.''

It is this unfolding policy issue that motivates the research question of this article: \emph{What would happen to data quality under a mandatory cyber data sharing regime for insurers?} To answer it, a game-theoretic model is constructed where cyber insurers interact on a Cournot oligopoly market, but are uncertain about their (and their competitors') production costs, i.e., the true costs of the cyber incidents underwritten. When forced to share what information they do have, they cannot refuse, but they can choose whether to invest in improving the data quality of their own information, or just provide it as-is. The model can be seen as an attempt to formalize Insurance Europe's remark about sharing equal quantities and qualities of data---how would such sharing unfold? It should be stressed that both the OECD and EIOPA stop short of recommending mandatory cyber data sharing laws. Nevertheless, the question is implicitly on the table, and our investigation aims to bring one more perspective to this important issue.

The rest of the paper is structured as follows. In the next section, some related work is discussed, and the contribution is positioned with respect to this literature. In Section~\ref{model} the formal model is introduced, and the main results are shown in Section~\ref{equilibria}. We find mandatory data sharing changes the feasible Nash Equilibria (NE) and creates situations in which at most one of the two insurers invests in data quality, whereas both insurers would invest when information is not shared. The results are followed by a discussion of implications and conclusions in Section~\ref{discussion}.

\section{Related work \label{related}}

The general topic of cyber security information sharing is extensively addressed in the literature. A good starting point is the literature survey provided by \cite{skopik2016problem}, who offer a comprehensive and broad overview of legal, technical and organizational aspects. \cite{koepke2017cybersecurity} provides a more focused literature review on incentives and barriers, complemented by a survey of 25 respondents. Of particular interest in our context are the collaborative barriers related to ``a lack of reciprocity from other stakeholders or the problem of free-riders. This barrier category also includes the risk of sharing with rivals/competitors who may use the shared information to enhance their competitive position'' \cite[p.~4]{koepke2017cybersecurity}.

\hlbreakable{
Turning to formal game-theory, two classic treatments are offered by \cite{gal2005economic,gordon2003sharing}, who show 
information sharing may yield benefits to firms, but also can result in free-riding.
}
Later works typically find similar results, and address the question of what the incentives needed for cooperation look like. For example, \cite{naghizadeh2016inter} study the effects of \emph{repeated interactions}, focusing on the effectiveness of monitoring regimes to detect and punish non-cooperative behavior.  Similarly, in a series of papers \citeauthor{tosh2015evolutionary} study a game where players face the binary choice of either participating in a sharing regime or not to, including its evolutionary stability \cite[see e.g.][]{tosh2015evolutionary,tosh2017risk}.

Whereas the previous treatments consider symmetric players---the potential victims of cyber attacks---asymmetric games have also been studied. \cite{laube2016economics} devise a principal–agent model of mandatory security breach reporting to authorities (such as those mandated under the GDPR and the NIS directive). Assuming imperfect audits which cannot determine for certain whether the failure to report an incident is deliberate concealment or mere lack of knowledge, \citeauthor{laube2016economics} find that it may be difficult to enact the sanctions level needed for the breach notification law to be socially beneficial.

\hlbreakable{
Whereas the works mentioned above treat for-profit parties interested in sharing and receiving information there are also non-profit actors who can participate in such arrangements.
For instance, \citet{dykstra2022economics} analyze information sharing of unclassified cyber threat information by a government institution. Such non-profit institutions may share unclassified information in order to improve social welfare, rather than maximize their own profits.
}

\hlbreakable{For a fuller literature review of game theory models of cyber security information sharing, see~\citet{laube2017strategic}, who not only summarize the literature, but also systematize it using an illuminating unified formal model. However, in our context, it is important to note that their review does not include any articles investigating information sharing among insurers. Thus, whereas information sharing \emph{between the firms at risk} is a standard component in game theoretic models of cyber risk---models which often include insurance---information sharing \emph{between the insurers} underwriting the firms at risk has not yet been formally investigated using game theory, despite the policy attention described in Section~\ref{intro}.}

Our model is inspired in the seminal work by \citet{gal1986information}, which addresses information transmission in oligopolies. 
In her model, firms can share information about their production cost, which is unknown and different for each firm. \citet{gal1986information} finds that under Cournot competition, firms chose to share information, because they benefit when competing firms make an accurate estimation of their production cost. In our model the insurers have the same cost (e.g., they compete in the same market); hence, by sharing information a firm may reduce the uncertainty of the competitor (which is what we expect in an insurance market).

\hlbreakable{
As mentioned above, we have not found any work formally investigating cyber security information sharing \emph{among insurers}.}
Instead, the work that is most closely related to ours is a qualitative study by \cite{nurse2020data}, who explore data use by cyber-underwriters in general, and the feasibility  and  utility of a  `pre-competitive dataset' shared within the industry in particular. Such a dataset is in fact precisely what is ``encouraged'' by the \cite{oecd2020cyberdata} and \cite{eiopa2020cyber}. However, the idea was met with considerable skepticism by the 12 cyber insurance professionals who participated in the focus groups conducted by \cite{nurse2020data}. They were all concerned about the implications for competitiveness, asking why incumbents would jeopardize their advantage by sharing information with market entrants. Indeed, the very structure of such a dataset was deemed sensitive, as even proposal forms are considered proprietary\hlbreakable{, even though there are published studies based on such forms, see \citet{woods2017mapping}}. ``People are insanely protective,'' remarked one participant \citep[p.~6]{nurse2020data}.


\section{Market model}
\label{model}

We use a Cournot model to study situations in which two insurers compete in a market, given that the claims (risk level) is uncertain. Before giving the formal statement of the model, it is appropriate to discuss some of the modeling choices. First, the Cournot model is an \emph{oligopoly} model. Thus, on the one hand, competition is not perfect---insurers make profits, which they would not if competition drove marginal prices down to equal marginal costs \cite[pp.~180–181]{varian1992microeconomic}. To understand why competition is not perfect, recall that economies of scale and rigorous regulation raise barriers to entry, making it harder for new insurance companies to challenge the incumbents. On the other hand, insurers are not monopolists who can raise prices arbitrarily---there is competition even among oligopolists. For cyber insurance, this is confirmed by several studies: \citet[p.~3]{nurse2020data} speak of ``an extremely competitive cyber insurance market'' and  \citet[p.~27]{woods2019does} fear that ``Competitive pressures drive a race to the bottom in risk assessment standards''. Furthermore, the Cournot model is not an uncommon choice for modeling general (non-cyber) insurance markets \citep[see, e.g.,][]{gale2002competitive,wang2003nightmare,cheng2008can,gao2016modeling}.

Second, production costs are uncertain---insurers do not know beforehand how much it will cost to produce their product, i.e. how large the indemnities owed will be. This reflects the uncertainty about cyber risk and lack of actuarial pricing described in Section~\ref{intro}: insurers underwriting cyber risks are uncertain about those risks.

Third, these uncertain production costs are assumed to be the same for all the market competitors. 
This reflects the interdependency of cyber risk described in Section~\ref{intro}: for an insurer. \hlbreakable{More precisely, cyber risk can only be managed up to a point by practices such as insuring customers in different geographical locations or from different industries. While such practices are effective against incident causes such as an outage at a payment service provider servicing a market of just one or a few countries, they are ineffective against other risks, such as the Heartbleed \citep[see, e.g.,][]{zhang2014analysis} or Log4J \citep[see, e.g.,][]{srinivasa2022deceptive} vulnerabilities, or prolonged outages at major cloud service providers \citep{Lloyds2018Cloud}. It is these risks---the ones that are difficult to manage---that is our concern here. With respect to these risks, thus, insurers can be seen} as essentially picking and insuring insureds from the very same set of eligible firms (with some firms being excluded by all insurers using similar rules-of-thumb). Thus, while the \emph{outcomes} of claims in a particular year will certainly differ, it is not unreasonable to model these outcomes with random variables representing production costs being the same for all market actors. Indeed, such an  assumption---in one form or another---is implicit in the entire discussion about data sharing.

Fourth, the uncertainties are modeled using normal random variables. \hlbreakable{The immediate rationale for this assumption is that it allows analytic calculations of conditional random variables. Therefore, it is} almost always used in the extant literature on uncertainties on oligopoly markets \citep[most prominently][ and the secondary literature citing her]{gal1986information}. \hlbreakable{However, this should be seen as a convenient mathematical approximation, not an empirical claim. Indeed, the literature on the statistics of cyber risk instead typically suggests more heavy-tailed distributions \citep[see, e.g., the review by][]{woods2021systematization}, and our model does not question that. This approximation is further discussed in Section~\ref{normal-discussion}.}

Turning to the formal model, in the Cournot competition two firms select their production levels,\footnote{It may seem unrealistic that production levels have to be chosen this way. After all, an insurance policy differs from physical goods and is not subject to the same production constraints. However, this is too simplistic---an insurer cannot scale production arbitrarily fast. Even if the constraints are not identical to those of physical production, important constraints are indeed imposed by, e.g., the ability to hire underwriters and claims managers, by access to capital (whether from investors or from bank loans), by the capacity of the brokers who act as middlemen on the cyber insurance market \citep[see, e.g.,][]{franke2017cyberinsurance,woods2019does}, and by regulation.} which determine the market price of their goods.
Let  $\mathcal{P}=\{1, 2\}$ be the set of firms and the real quantity $q_i\in \mathbb{R}$ their production, for $i\in\mathcal{P}$. 
In this case $q_i$ represents the number of policies offered.
We define the inverse demand function, i.e., the unitary price of a product as
\begin{equation}\label{eq:p_i}
 p_i(q_i, q_j) = a - b q_i - d q_j,
\end{equation}
where  $a, \, b, \, d> 0$. The value $p_i$ represents the premium, i.e., the payment that the insurer $i$ receives.

We assume that each insurer has a linear production cost $q_i C_i$, where $C_i$ is the marginal cost (the claims of each policy).
%
%
%
%
For simplicity we assume that the insurers offer identical products ($b=d$), \hlbreakable{allowing the price to be written as a function of the total production $p(q_i, q_j) = p(q_i+q_j) = a - b(q_i+q_j)$}. We also assume that the insurers have the same marginal production costs ($C_i=C$), an unknown value with distribution 
$C \sim \mathcal{N}(0, \sigma)$, where $\sigma$ is the uncertainty about the production cost.\footnote{While variance is often denoted  $\sigma^2$, for simplicity, we adhere to the notation of \cite{gal1986information} and denote it as $\sigma$.} \hlbreakable{Note that the assumption that $C$ has mean zero does not affect the results, but considerably simplifies the exposition.}\footnote{\hlbreakable{
Consider the marginal cost $C' = \bar c + C$, where $\bar c\in\mathbb{R}$ and $C\sim \mathcal{N}(0, \sigma)$ represent fixed and variable components. 
With an inverse demand function of the form $a' - b q_i -d q_j$ the profit can be rewritten as 
\begin{align}
q_i (a'  - b q_i - d q_j) - q_i (\bar c + C) =
q_i (a' - \bar c - b q_i - d q_j) - q_i C =
q_i p_i(q_i + q_j) - q_i C   
\end{align}
The last step results by making $a=a'-\bar c$.
We can set the parameter $a'$  large enough to guarantee that $a>0$. 
}}

\hlbreakable{The use of a single random cost variable merits some additional discussion. Importantly, this means that insurers do not benefit from the law of large numbers, as they would if instead there was rather a sum of $q_i$ single random cost variables---one per insured. To understand why this is reasonable, recall the discussion above about the fact that cyber risk can only be managed up to a point  by practices such as insuring customers in different geographical locations or from different industries. Some risks remain, namely, the ones stemming from irreducible interdependence, as discussed in Section~\ref{intro}. These risks---e.g., the risks that all of the insureds are hit by something like Heartbleed or Log4J---are precisely the ones that have prompted the policy interest in cyber insurance information sharing, i.e., the risks that are our concern here. These risks are well modeled by the use of a single random cost variable, and this is what the model is designed to reflect. Of course, this is not to deny that the law of large numbers works well for \emph{other} risks, including (some) other cyber risks, and that this is a cornerstone of insurance. But the model developed here aims to reflect precisely the interdependent cyber risks that cannot be tamed by the law of large numbers.}

In our model each insurer conducts a risk assessment and finds a noisy signal about the claims (i.e., the production cost), denoted $Z_i = C + E_i$, where the noise
$E_i\sim\mathcal{N}(0, m_i)$ is independent from the cost $C$.
Here $m_i$ represents the uncertainty inherent in the signal.
We assume that $Z_i$ is a private signal that depends on investments to improve the risk assessment process and its output. 
\hlbreakable{We do not explicitly model the mechanisms by which $Z_i$ can be improved. However, it is clear that many possibilities exist, ranging from better security audits before underwriting clients, to continuous SIEM-like monitoring of clients' systems, to improving DFIR (Digital Forensics and Incident Response) processes once incidents occur. It is equally clear that such possibilities entail costs.} 
Note that as opposed to \citeauthor{gal1986information}, we do not assume that firms deliberately garble any information. We do, however, assume that data quality is a real issue, and that it may be low in the absence of deliberate and costly efforts to improve it. Recall the results from the study of Swedish NIS reports: ``Making the most of the reporting requires additional quality assurance mechanisms'' \citep{Franke2021NIS}.
Concretely, an investment $h_i\geq 0$ leads to an uncertainty
\begin{equation}\label{eq:m_i}
 m_i = m_0 \alpha^{-h_i},
\end{equation}
where $m_0>0$ represents the uncertainty without any investments (e.g., when the assessment is made using publicly available data\footnote{
\hlbreakable{
Public reports, like the studies conducted by NetDiligence or the Ponemon Institute, may offer some information about the cyber risks of different industries. 
In addition, information sharing may take place among business partners, e.g., between insurance firms and their reinsurance providers and/or third parties that offer technical support.
}
}) and $\alpha >1$ represents the efficacy reducing noise.
Again, we assume that these parameters are equal for both insurers.
From \cref{eq:m_i} the investment needed to have an uncertainty $m_i$ is
\begin{equation}\label{eq:x_i}
 h_i(m_i) = \log_\alpha(m_0/m_i),
\end{equation}
which is a concave with respect to the uncertainty level $m_i$.
\cref{eq:x_i} implies that 
it's prohibitively expensive to have no uncertainty ($m_i=0$). 
Now, the profit of firm $i$, denoted $\pi_i(q_i, q_j)$, is equal to its income minus both production and risk assessment costs
\begin{equation}\label{eq:cournot}
 \pi_i(q_i, q_j) = q_i \cdot p_i(q_i + q_j) - q_i \cdot C - h_i(m_i).
\end{equation}

\subsection{Game formulation}

\begin{figure}
    \centering
    \includegraphics[width=.8 \textwidth]{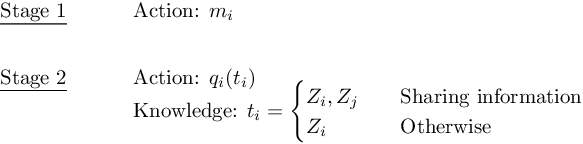}
    \caption{Our Cournot model has two decision stages: 1) insurers decide the uncertainty of their data $m_i$ (and investment in data quality); 2) insurers use their available information to choose their production level (policies offered) $q_i$.}
    \label{fig:stages}
\end{figure}

In our Cournot model the insurers make two decisions at different stages (see \cref{fig:stages}):
\begin{enumerate}
 \item In the first stage insurers make investments and commit to an information sharing policy.
This is equivalent to selecting the uncertainty level $m_i$, which in turn determines the investment $h_i(m_i)$.
We assume that 
the sharing policy is defined before the game starts.

\item In the second stage the marginal production cost $C$ is realized and each firm gets an estimate  $Z_i$.
Then the information transmission takes place and each firm uses the information available, represented as $t_i$, to select the production quantity $q_i(t_i)$.
The information available is $t_i=(Z_i, Z_j)$ when insurers share their cost estimations, otherwise it is $t_i=(Z_i)$.

\end{enumerate}

Thus, the strategy of each firm has the form $(m_i, q_i(t_i))$, which satisfies the subgame perfect equilibria if it is a Nash Equilibrium (NE) in each stage of the game.
We start by analyzing the second stage game to determine the production $q_i(t_i)$ (see \cref{fig:stages}). With this result we build the game in the first stage and then formulate the problem of selecting the optimal uncertainty (noise level) of the data $m_i$.
Some of the results in this section resemble the findings of \citet{gal1986information}, because the Cournot model there is similar (at a high level) to ours; however, the precise solution for our model is different.
We find the precise profit function for each scenario in the next section.









\subsection{Second stage (low level) game} \label{sec:game_2_risk_neutral}

In this stage $m_i$ and $m_j$ are given and each firm chooses a production $q_i(t_i)$ that maximizes their expected profits given the available information $t_i$ (\cref{fig:stages}). 
We define the game in the second stage as
\begin{equation}\label{eq:G2}
\mathcal{G}_2 = \langle \mathcal{P}, (S_i)_{i\in \mathcal{P}}, (W_i)_{i\in \mathcal{P} } \rangle,
\end{equation}
where $\mathcal{P}$ is the set of players, $S_i=\mathbb{R}$ is the strategy space, and $W_i$ is the payoff function of the i$\th$ firm, which in this case corresponds to the expected profit in a Cournot competition (see \cref{eq:cournot}) given the signal $t_i$. 
The following result shows the form of $W_i$ as a function of the optimal production $q_i(t_i)$.

\begin{lemma}\label{lemma:W_i}
 The utility of the  game $\mathcal{G}_2$ defined in \cref{eq:G2} is 
\begin{equation}
 W_i(q_i, q_j) = b q_i^2(t_i) - h_i(m_i),
\end{equation}
where the optimal production $q_i(t_i)$ satisfies
\begin{equation}\label{eq:q_i_b}
 q_i(t_i) = \frac{ a - b\E \{ q_j(t_j)| t_i\}  - \E_C \{  C | t_i \} }{ 2 b}.
\end{equation}
\end{lemma}
\begin{proof}
The expected profit in a Cournot competition (see \cref{eq:cournot}) is
\begin{equation}
 W_i(q_i, q_j) = \E \{ \pi_i(q_i, q_j) | t_i  \}.
\end{equation}
The expectation is made with respect to the unknown parameters, such as the cost $C$ (and the signal $Z_j$ when the firms do not share information). Thus,
\begin{equation}\label{eq:w_i}
 W_i(q_i, q_j) =  q_i(t_i) \left( a - b q_i(t_i) - b \E \{ q_j(t_j) | t_i\} -  \hat C_i \right)  - h_i(m_i) 
\end{equation}
 where $\hat C_i = \E_C \{  C | t_i \}$ is an estimation of the production cost and
$\E \{ q_j(t_j) | t_i\}$ is the estimated production of the adversary, given the available observation $t_i$.
The optimal production must satisfy the following first order condition (FOC)
\begin{equation}\label{eq:W_i_1st_d}
 \frac{ \partial W_i }{ \partial q_i } = a - 2 b q_i(t_i) - b   \E \{ q_j(t_j)| t_i\}  - \hat C_i = 0.
\end{equation}
In this case
the optimal production is unique, since the expected profit is concave with respect to $q_i$:
\begin{equation}\label{eq:W_i_2nd_d}
 \frac{ \partial^2 W_i }{ \partial q_i^2 } = - 2 b < 0.
\end{equation}
Now, from \cref{eq:W_i_1st_d} the optimal production satisfies
\begin{equation}\label{eq:q_i}
 q_i(t_i) = \frac{ a - b\E \{ q_j(t_j)| t_i\}  - \hat C_i }{ 2 b}.
\end{equation}
Replacing \cref{eq:q_i} in \cref{eq:w_i} we obtain
\begin{align}
 W_i(q_i, q_j) & = q_i(t_i) \left( 2 b q_i(t_i) - b q_i(t_i)  \right)  - h_i(m_i) \\
  & = b q_i^2(t_i) - h_i(m_i).
\end{align}
\end{proof}

\subsection{First stage (upper level) game} 
 
 We define the game in the first  stage (see \cref{fig:stages}) as 
\begin{equation}\label{eq:G1}
\mathcal{G}_1 = \langle \mathcal{P}, ([0, m_0])_{i\in \mathcal{P}}, (J_i)_{i\in \mathcal{P} } \rangle,
\end{equation}
where insurers can select an uncertainty level  $m_i\in[0, m_0]$ and the utility $J_i$ is the expected profit, given that firms choose $q_i(t_i)$ in the second stage (see \cref{lemma:W_i}).
In this stage $t_i$ hasn't been realized; hence, the optimal production $q_i(t_i)$ can be seen as a random variable.
The following result shows the form of $J_i$.
\begin{lemma}\label{lemma:J_i}
The utility of the game $\mathcal{G}_1$ defined in \cref{eq:G1} is
\begin{equation}\label{eq:J_i_b}
 J_i(m_i, m_j) =  b \left( \mu_i^2 + \gamma_i  \right)  - h_i(m_i),
\end{equation} 
 where $\mu_i = \E[q_i(t_i)]$ and $\gamma_i=\Var[q_i(t_i)]$ are the first and second moments of the optimal production $q_i(t_i)$, respectively  (see \cref{eq:q_i_b}).
\end{lemma}
\begin{proof}
In the first stage the utility (expected profit) is
\begin{equation}\label{eq:J_i}
 J_i(m_i, m_j) = \E_{Z_i, Z_j} \left\{ W_i(q_i(t_i), q_j(t_j) )  \right\} =  b \E_{Z_i, Z_j} \left \{ q_i^2(t_i) \right \} - h_i(m_i).
\end{equation}
The expectation is with respect to the signals $Z_i$ and $Z_j$, which haven't been realized in  the stage. 
Let us define $\mu_i =  \E[q_i(t_i)]$ and $\gamma_i = \Var[q_i(t_i)] $.
Since $\Var [X] = \E[X^2] - \E[X]^2 $,  then $ \E[X^2] =\Var [X] + \E[X]^2$.  Hence, the expected profit in \cref{eq:J_i} can be written as \cref{eq:J_i_b}.
\end{proof}

\begin{remark}\label{remark:equilibria}
We cannot guarantee that $\mathcal{G}_1$ has a Nash Equilibrium (NE). 
Also, rather than finding the precise NE of $\mathcal{G}_1$, we focus on finding the conditions in which investing in risk assessment is feasible or not. Thus, we classify the possible NE in the following categories
\begin{itemize}
    \item Both insurers invest (but different amounts): $(m_i, m_j)$, with $m_i, m_j\in(0, m_0)$ and $m_i \neq m_j$.
    \item Both insurers invest the same amount: $(m, m)$, with $m\in(0, m_0)$.
    \item Only one insurer invests in risk assessment: $(m, m_0)$, with  $m\in(0, m_0)$.
    \item Neither insurer invests: $(m_0, m_0)$.  
\end{itemize}
\end{remark}

\subsection{Cost Estimation}

Since both the cost $C$ and the noise $E_i$ are normally distributed, the sample $Z_i=C+E_i$ is also normally distributed: $Z_i\sim \mathcal{N}(0, \sigma+m_i)$. 
This property makes it easier to find the closed form expressions of variables of interest. For instance, the estimation of the cost conditional to the sample $Z_i$ is
\begin{equation}\label{eq:hat_c_i}
 \bar C_i =  \E_C \{ C | Z_i \} 
 = \delta_i Z_i
\end{equation}
with $\delta_i = \frac{\sigma}{\sigma+m_i} $. 
Moreover, we can find the expected cost given two observations, $Z_i$ and $Z_j$, using the multivariate normal distribution
 \begin{equation}\label{eq:hat_C}
 \bar C_{ij} = \E_C \{  C | Z_i, Z_j \} = 
 \frac{1}{ k_0 } \left( \sigma m_j Z_i + \sigma m_i Z_j \right)  ,
 \end{equation}
where 
$k_0 = \sigma m_i + \sigma m_j + m_i m_j$.
\cref{fig:bivariate} shows the Bivariate distribution of $Z_i$ and $Z_j$ with parameters $m_i=m_j=2$ and $\sigma=4$.

The samples $Z_i$ and $Z_j$ are correlated through the cost $C$; hence, $\Cov(Z_i, Z_j)=\sigma$. 
For this reason the insurer $i$ can estimate the sample of it's adversary $Z_j$ as
\begin{equation}\label{eq:estimate_zj}
  \E \{ Z_j | Z_i \} = \E \{ C + E_j | Z_i \} 
  = \E \{ C | Z_i \} .
\end{equation}
Note that sharing information creates a conflict, because although insurers may benefit, they may also help the competing insurer.

%
\cref{fig:ex_cost_dist} shows an example of the cost distribution given some observations. In this case each additional observation reduces the uncertainty of the cost (variance).
Note that the insurers use $\E[C|Z_i]$ to make decisions in the second stage, rather than the rv $(C|Z_i)$
depicted in \cref{fig:ex_cost_dist}.

\begin{figure}
 \centering
  \scalebox{.6}{
  \input{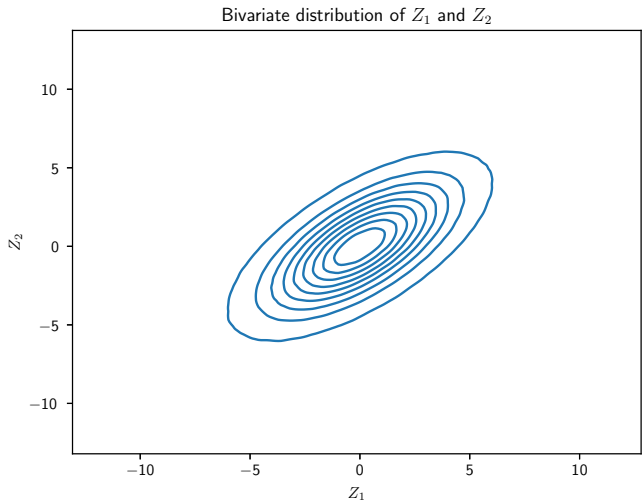}
   }
 \caption{Bivariate distribution of the samples $Z_i$ and $Z_j$. Since the samples are positively correlated, sharing information benefits an insurer but also may benefit competitors.}
 \label{fig:bivariate}
\end{figure}

\begin{figure}
 \centering
  \scalebox{.6}{
  \input{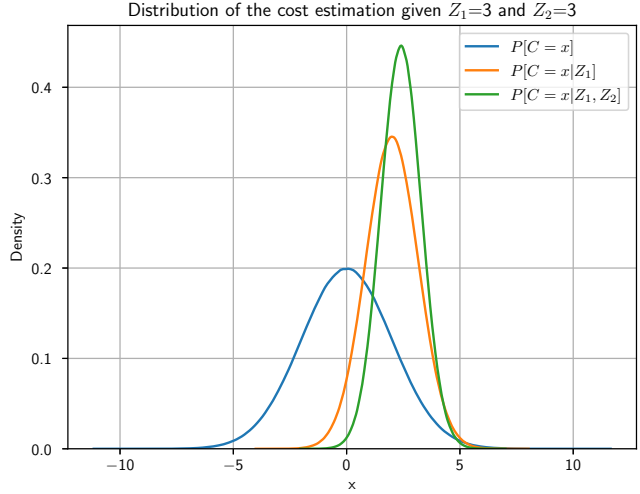}
   }
 \caption{Example of the distribution of the cost estimation with observations $Z_i$ and $Z_j$. Each additional observation reduces the uncertainty (variance) of the estimation.}
 \label{fig:ex_cost_dist}
\end{figure}

\section{Analysis of the market's equilibria} \label{equilibria}

In this section we find the optimal production $q_i(t_i)$ and the utility function in the upper level game $J_i$ for each scenario (sharing and non-sharing information). Then we find the possible equilibria of the game
and illustrate them with examples.


\subsection{Market with information sharing}

In this case the firms share their private information $Z_i$. Thus, they have the same signal $t_i = t_j = t = (Z_i, Z_j)$, and therefore,
make the same cost estimation  $\hat C_i = \hat C_j = \hat C$, where $\hat C =  \bar C_{ij}$ (see \cref{eq:hat_C}).
Moreover, since 
firms have the same characteristics and possess the same information, there is no uncertainty about the production  of the adversary, because $\E \{ q_j(t_j) | t_i\} = \E\{ q_j(t) | t \} = q_j(t)$. 
In addition,  
they produce the same quantity
$q_i(t)=q_j(t)$. Thus, from \cref{eq:q_i} we get
\begin{equation}\label{eq:q_i_sharing}
 q_i(t) = \frac{1}{3 b} ( a - \hat C ).
\end{equation}


Now, recall that the cost $C$,  $E_i$, and $E_j$ are independent normal random variables.
Thus, in the first stage, when neither $C$, $Z_i$ nor $Z_j$ have been realized, we can see $\hat C$ as a random variable (see \cref{eq:hat_C})
 \begin{equation}
  \hat C 
 = k_c C + k_i E_i + k_j E_j,
 \end{equation}
 where
$k_c =  \frac{ \sigma(m_i+m_j) }{k_0}$,
$k_i = \frac{ \sigma m_j }{k_0}$, and
$k_j = \frac{ \sigma m_i }{k_0}$, with $k_0=\sigma (m_i + m_j )  + m_i m_j $.
 Thus, the cost estimation $\hat C$ is normally distributed
$\hat C \sim \mathcal{N}( 0 , \hat \sigma )$
with
\begin{equation}
\hat \sigma  = \frac{\sigma^2}{k_0} (m_i+m_j).
\end{equation}

Now, from \cref{eq:q_i_sharing}
 the optimal production $q_i$ can be seen as a random variable
$q_i \sim \mathcal{N}( \frac{a}{3b} , \frac{\hat \sigma}{9 b^2} )$.
Therefore, the profit of the game $\mathcal{G}_1$ (see \cref{eq:J_i_b}) is
\begin{align}\label{eq:J_i_sharing}
  J_i(m_i, m_j) & = \frac{1}{9b} (  \hat \sigma + a^2 ) - h_i(m_i) \\
  & = \frac{1}{9b} \left( \frac{ \sigma^2 (m_i+m_j) }{  \sigma (m_i + m_j) + m_i m_j }  + a^2 \right) - h_i(m_i).
\end{align}


The following result shows that only one firm may invest in risk assessment.
\begin{proposition}\label{lemma:eq_sharing}
 A duopoly in which insurers share information can have two types of Nash equilibria (but only one of this scenarios can occur in a game): 
 \begin{itemize}
 \item Neither firm invests ($m_i=m_j=m_0$) if $m_0 < \frac{36 b}{\log \alpha}$
or if 
 \begin{equation}
 m_0 > \frac{36 b}{\log \alpha} \quad \text{and} \quad \sigma < \tilde \sigma = \frac{36 b m_0}{ m_0 \log \alpha - 36 b }.
 \end{equation}
 
 \item  Only one insurer invests 
 (e.g., $m_i\leq m_0$ and $m_j=m_0$) if 
   \begin{equation}
   m_0 \geq \frac{36 b}{\log \alpha} 
  \quad \text{and} \quad  \sigma \geq \hat \sigma = 
   \frac{m_0}{\Gamma-2} ,
\end{equation}
where $\Gamma^2=\frac{m_0 \log \alpha}{9 b}$.
 \end{itemize}
\end{proposition}
The reader can find the proof of this and the following results in the appendix.

\cref{fig:eq_area_sharing} illustrates the possible equilibria for different values of $\sigma$ and $m_0$ (see \cref{lemma:eq_sharing}). 
\cref{fig:examples_eq_sharing} shows examples with the two possible NE, namely $(m_0, m_0)$ and $(\hat m, m_0)$.
In this case, free-riding can occur because at most one insurer invests in data quality (when $m_0$ is large, i.e., when the data quality is low without any investment)
.
In these examples we use the following parameters:
$a = 10$, $b = 1$, $\alpha = 3 $, and $m_0 = 1.5 \frac{36 b }{ \log \alpha}$. 
Moreover, for the NE examples we use
$\sigma = 1.2 \hat \sigma $ (\cref{fig:ex_a_sharing}) and 
$\sigma = 0.9 \tilde \sigma $  (\cref{fig:ex_b_sharing}).

\begin{figure}
 \centering
  \scalebox{.6}{
  \input{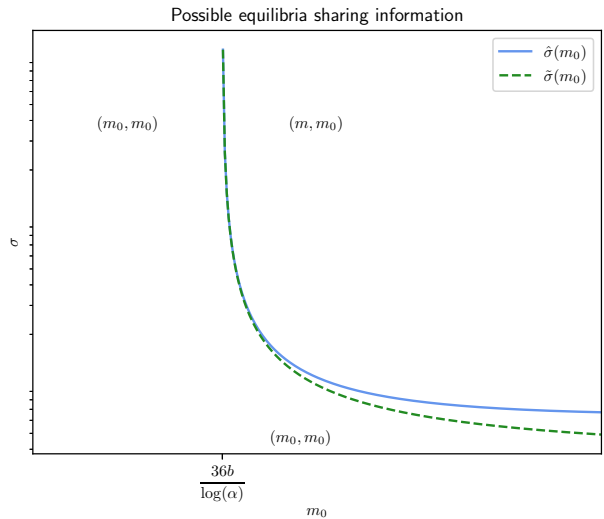}
   }
\caption{When the insurers share information at most one insurer invests in data quality.}
 \label{fig:eq_area_sharing}
 \end{figure}

\begin{figure}
 \centering
    \begin{subfigure}[b]{1\textwidth}
        \centering
        \scalebox{.6}{
        \input{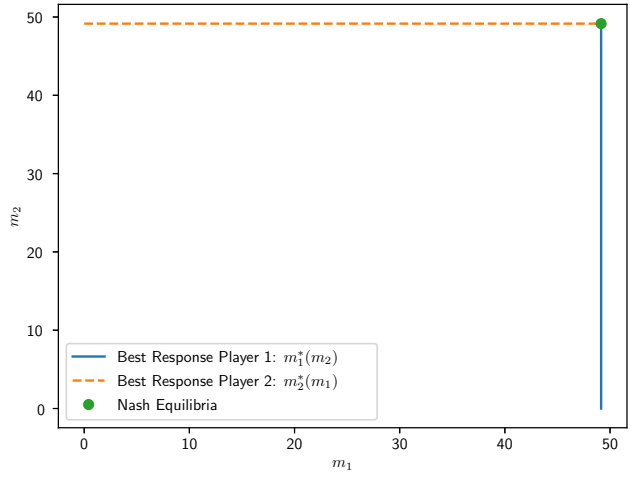}
        }
        \caption{Neither insurer invests in risk assessment.}
        \label{fig:ex_a_sharing}
    \end{subfigure}

    \begin{subfigure}[b]{1\textwidth}
        \centering
        \scalebox{.6}{
        \input{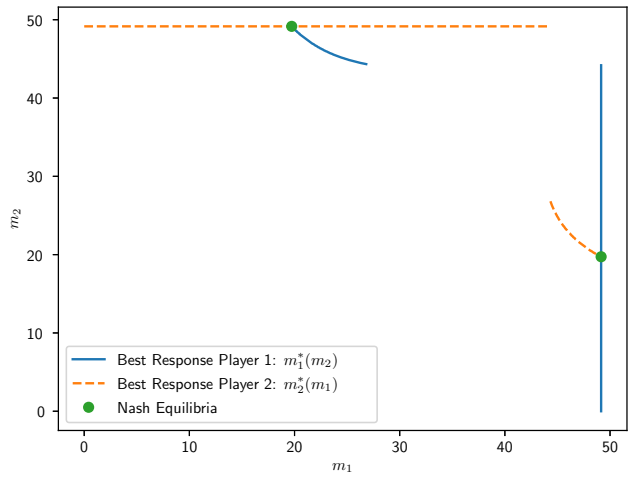}
        }
        
        \caption{Only one insurer invests in risk assessment.}
        \label{fig:ex_b_sharing}
    \end{subfigure}
 \caption{Best response of insurers in a market where sharing  information is enforced.}
 \label{fig:examples_eq_sharing}
\end{figure}

\subsection{Market without information sharing}

In this case the information available is $t_i = (Z_i)$.
From \cite{gal1986information, radner1962team} the decision rules must be affine in the vector of observations (in this case $Z_i$), therefore, 
\begin{equation}\label{eq:q_i_c1_a}
 q_i(t_i) = q_i(Z_i) = \alpha_0^i+ \alpha_1^i Z_i,
\end{equation}
for some constants $\alpha_0^i$ and $\alpha_1^i$.
From \cref{eq:q_i_c1_a} we can find the expected demand of the adversary given the information $t_i$
\begin{equation}\label{eq:exp_q_j}
\E\{ q_j(t_j) | t_i \} = \E\{ q_j(Z_j) | Z_i \} =  \alpha_0^j + \alpha_1^j \E \{ Z_j | Z_i \}
\end{equation}
We can use  \cref{eq:estimate_zj} to rewrite \cref{eq:exp_q_j} as
\begin{equation}
 \E\{ q_j(t_j) | t_i \}
= \alpha_0^j + \alpha_1^j \delta_i Z_i.
\end{equation}
Replacing \cref{eq:hat_c_i} and \cref{eq:exp_q_j} in \cref{eq:q_i} we obtain
\begin{equation}\label{eq:q_i_c1_b}
 q_i(Z_i) = \frac{ a - \delta_i Z_i - b (\alpha_0^j + \alpha_1^j \delta_i Z_i) }{ 2b }
\end{equation}
\cref{eq:q_i_c1_a} and \cref{eq:q_i_c1_b} are equivalent, therefore the coefficients $\alpha_0^i$ and $\alpha_1^i$ must  satisfy
\begin{equation}\label{eq:coef_q_i}
 \alpha_0^i = \frac{a-b \alpha_0^j }{2b}
\quad \text{and}  \quad
\alpha_1^i = -\delta_i \frac{1+b \alpha_1^j}{2b},
\end{equation}
which have the following solution 
\begin{equation}
 \alpha_0^i = \frac{a}{3b}
\quad \text{and}  \quad
\alpha_1^i = \frac{ \sigma ( 2(\sigma + m_j) - \sigma) }{b(\sigma^2 - 4(\sigma+m_i)(\sigma + m_j))}.
\end{equation}

Now, in the first stage the optimal production $q_i(t_i)$ can be seen as a random variable.
Since $Z_i$ is normal, the production is also normal (see \cref{eq:q_i_c1_a})
\begin{equation}
 q_i(Z_i) \sim \mathcal{N}( \alpha_0^i, (\alpha_1^i)^2 (\sigma+m_i) ) .
\end{equation}
Then, the  profit of $\mathcal{G}_1$ is
\begin{align}
 J_i(m_i, m_j) & = b( (\alpha_0^i)^2 + (\alpha_1^i)^2 (\sigma+m_i) ) - h_i(m_i) \\
  & = 
  \frac{a^2}{9b} +
 \frac{ \sigma^2 ( 2(\sigma + m_j) - \sigma)^2 (\sigma +m_i) }{b(\sigma^2 - 4(\sigma+m_i)(\sigma + m_j))^2} - \log_\alpha (m_0/m_i)
\end{align}


The following result states that the market can have three types of equilibria. Unlike the previous case, not sharing information creates the conditions to have both insurers investing in data quality.
\begin{proposition}\label{lemma:eq_non_sharing}
A duopoly in which insurers do not share information can have 
the following equilibria
\begin{itemize}
    \item $(m_0, m_0)$ is the only NE if 
    $m_0\leq \frac{27b}{5 \log \alpha}$

    \item $(m_0, m_0)$ is the only feasible NE 
    if 
    \begin{equation}
     m_0 \geq \frac{27 b}{5 \log \alpha} \quad \text{and} \quad 
    \sigma \leq \acute \sigma=\frac{2 m_0}{ \tilde \gamma^{1/2} - 3 },
    \end{equation}
    where 
    $\tilde \gamma = \frac{5 m_0 \log \alpha}{3 b}$.
    
    \item $(\hat m, m_0)$ and  $(\hat m, \hat m)$ are feasible NE if
    \begin{equation}
     m_0 > \frac{9 b }{ \log \alpha} \quad \text{and} \quad 
     \sigma \geq \breve \sigma = \frac{2 m_0}{ \hat \gamma^{1/2} - 3 }, 
    \end{equation}
    where
    $\hat \gamma = \frac{m_0 \log \alpha}{b}$.
    \end{itemize}
\end{proposition}

\cref{fig:eq_regions_ns} shows the feasible equilibria (when insurers do not share information) for different values of $\sigma$ and $m_0$.
Investments in data quality occur only when $m_0$ is large, that is, when the data quality is low without any investment.
\cref{fig:examples_eq_nonsharing} shows two examples of the possible equilibria in the market without sharing information.
In particular, \cref{fig:ex_b_nonsharing} shows that, unlike in the market that with information sharing enforced, one or both firms can invest in data quality in the equilibria.
In these examples $m_0 = 1.5 \frac{36 b}{\log \alpha} $ (same as the example in \cref{fig:examples_eq_sharing}).
\cref{fig:ex_a_nonsharing} has $\sigma=0.9 \tilde \sigma$ and 
\cref{fig:ex_b_nonsharing} $\sigma=1.2 \hat \sigma$.

\begin{figure}
 \centering
  \scalebox{.6}{
  \input{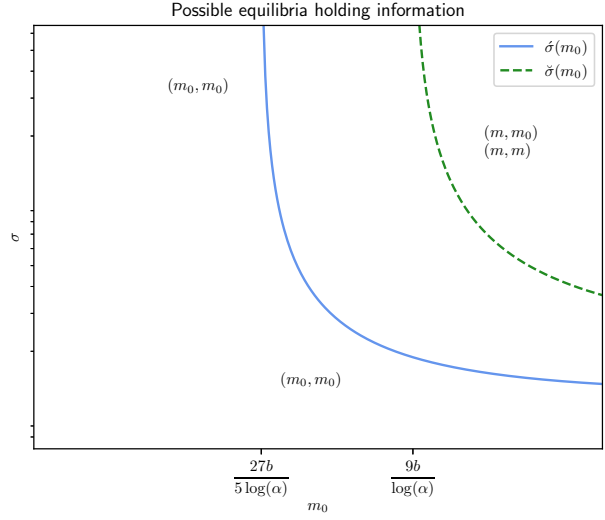}
   }
\caption{Without information sharing the insurers make symmetric or a single-sided investments.}
 \label{fig:eq_regions_ns}
\end{figure}

\begin{figure}
 \centering
    \begin{subfigure}[b]{1\textwidth}
        \centering
        \scalebox{.6}{
        \input{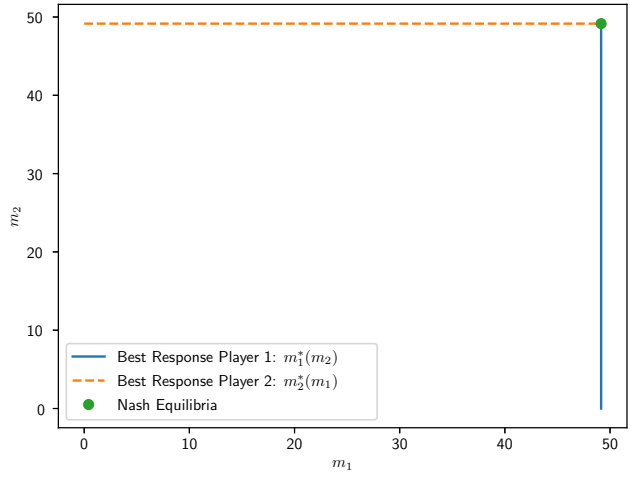}
        }
        \caption{Neither insurer invests in risk assessment.}
        \label{fig:ex_a_nonsharing}
    \end{subfigure}

    \begin{subfigure}[b]{1\textwidth}
        \centering
        \scalebox{.6}{
        \input{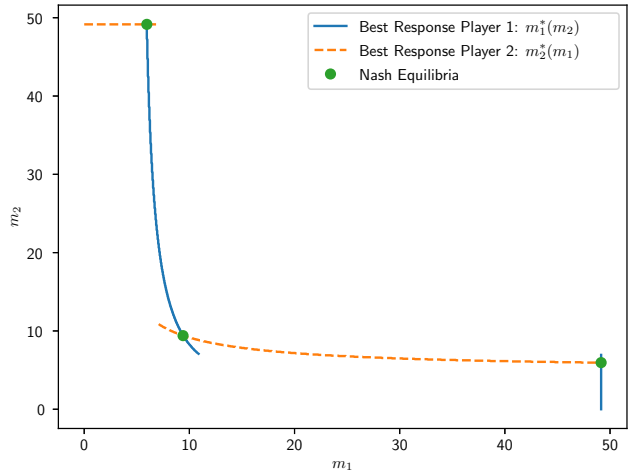}
        }
        
        \caption{Both insurers invest the same amount of resources in risk assessment or only one firm invests.}
        \label{fig:ex_b_nonsharing}
    \end{subfigure}
 \caption{Best response of insurers when they do not share information.}
 \label{fig:examples_eq_nonsharing}
\end{figure}

\modification{
\section{Relaxing the normal approximation}
\label{normal-discussion}

In this section we argue that the previous results are also valid---to some extent---when we consider extreme events in the costs, i.e., cost distributions with heavier right tails than the normal distribution, as is often the case with cyber risks.
Let us consider a cost distribution with pdf
\begin{equation}
\Prob[C = x] = 
\begin{cases}
w_1 f_1(x) & \quad \text{ if } -\infty < x \leq x_0 \\
w_2 f_2(x) & \quad \text{ if } x_0 \leq x < \infty
\end{cases}
\end{equation}
where $w_1, w_2 \geq 0$, $w_1+w_2=1$,  $f_1$ is a normal pdf,  $f_2$ a Generalized Pareto Distribution (GPD). In this case, claims that exceed the threshold $x_0$ are modeled with the GPD.
Now, an estimation of the cost given some information $t_i$ is
\begin{equation}
\E[C|t_i] = \int_{-\infty} ^ {\infty} x \Prob[C=x|t_i] dx.
\end{equation}
Let us decompose the estimation into two terms, one corresponding to the most frequent events and one to the tail
\begin{align}
\E[C|t_i] &
= \int_{-\infty}^{x_0} x w_1 f_1(x | t_i ) dx + \int_{x_0} ^ {\infty} x w_2 f_2(x|t_i) dx \\
& = \bar C_i + \epsilon_i.
\end{align}
Here $\bar C_i$ and $\epsilon_i$ represent the lower and upper estimates of the cost.
We assume that the first term $\bar C_i$ is close to the cost estimation assuming a normal distribution.
Intuitively, estimations assuming a normal distribution ignore the contribution of the tail.

The optimal amount of policies issued by each insurer---see \cref{eq:q_i_b}---then becomes\footnote{
We can use \cref{lemma:W_i,lemma:J_i} because they are independent of the distributions of $C$ and $t_i$.
}
\begin{equation}\label{eq:q_i_tail_}
q_i(t_i) = \frac{a - b \E\{ q_j(t_j)|t_i \} - \bar C_i - \epsilon_i(t_i)}{2 b}.
\end{equation}
Note that the optimal production is lower when we consider the costs from the tail.
Now, let us express the optimal production as
\begin{equation}\label{eq:q_i_tail}
q_i(t_i) = \tilde q_i(t_i) - \delta_i(t_i)
\end{equation}
where $\tilde q_i(t_i)$ is the optimal production assuming a normal distribution.
Substituting \cref{eq:q_i_tail} into \cref{eq:q_i_tail_} we obtain
\begin{equation}
q_i(t_i) = \frac{a - b \E\{ \tilde q_j(t_j) - \delta_j(t_j)|t_i \} - \bar C_i - \epsilon_i(t_i) }{2 b}.
\end{equation}
The previous expression can be rewritten as 
\begin{align}\label{eq:q_i_tail_b}
q_i(t_i) & = \frac{a - b \E\{ \tilde q_j(t_j) |t_i \} - \bar C_i }{2 b} - \frac{b \E\{ \delta_j(t_j)|t_i \} + \epsilon_i(t_i)}{2 b}\\
& = \tilde q_i(t_i) - \frac{b \E\{ \delta_j(t_j)|t_i \} + \epsilon_i(t_i)}{2 b}.
\end{align}
The last step results from using \cref{lemma:W_i} to get the optimal production when we estimate the cost assuming a normal distribution (i.e., using the cost estimate $\bar C_i$). 
From \cref{eq:q_i_tail} and \cref{eq:q_i_tail_b} conclude that
\begin{equation} \label{eq:delta_i}
\delta_i(t_i) = \frac{b \E\{ \delta_j(t_j)|t_i \} + \epsilon_i(t_i)}{2 b}.
\end{equation}
Through \cref{eq:delta_i} we express the impact of the tail in the optimal decision (in the second stage of the game). Now we can analyze how the tail affects the game in the first stage. Concretely, the utility function becomes (see \cref{lemma:J_i}) 
\begin{equation} \label{eq:J_i_tail}
J_i(m_i, m_j) = b( \E[\tilde q_i - \delta_i]^2 + \Var[\tilde q_i - \delta_i] ) - h_i(m_i).
\end{equation}
Note that $\delta_i$ depends on both $\epsilon_i$ and $\epsilon_j$. Thus, we argue that if the tail of the cost has finite mean and variance, then we can approximate  \cref{eq:J_i_tail} with
\begin{equation}
  \tilde  J_i (m_i, m_j) = b( \E[\tilde q_i ]^2 + \Var[\tilde q_i] ) - h_i(m_i),
\end{equation}
where $\tilde  J_i (m_i, m_j)$ represents the utility in the first stage that we obtain using the normal approximation.
If the tail of the cost has finite mean and variance,\footnote{Heavy tail distributions like the Lognormal distribution, sometimes used to model the size of cyber incidents \citep[][]{woods2021systematization}, satisfy  $\E[\epsilon_i]<\infty$ and $\Var[\epsilon_i]<\infty$.} then we can find some bounds $\ubar \phi_i, \bar\phi_i >0$ s.t.
\begin{align}
    J_i (m_i, m_j) \geq \tilde  J_i (m_i, m_j) - \ubar \phi_i \\
    J_i (m_i, m_j) \leq \tilde  J_i (m_i, m_j) + \bar \phi_i .
\end{align}
It follows from \cref{lemma:eNE} that the NE obtained assuming a normal distribution is close to the NE that we would obtain considering the tail of the distribution \citep[see the discussion of $\epsilon$-equilibria in][]{myerson1978refinements, fudenberg1986limit}. How close depends on the size of the tail. 
}
%

\section{Discussion and Conclusions \label{discussion}}

\cref{fig:comparison} compares the equilibria resulting with each information policy. 
When insurers share information, they tolerate data with poorer quality (i.e., large $m_0$) before starting to invest.
Also, forcing information sharing may reduce the investment in data quality of the risk assessments, because the NE changes from $(m, m_0)$ or $(m, m)$ to a NE where neither insurer invests $(m_0, m_0)$ (see region A).
Likewise, mandatory information sharing can restrict the NE to a situation of free-riding, because $(m, m)$ is not feasible (see region B).

\begin{figure}
 \centering
  \scalebox{.6}{
  \input{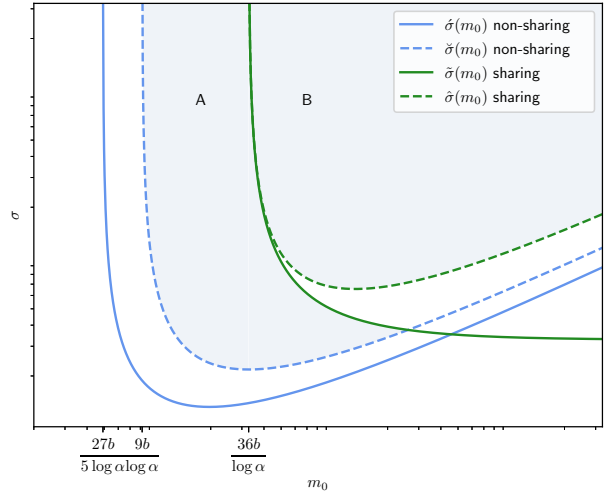}
   }
\caption{Regions where a change in the information sharing policy may change the NE. In region A the NE changes from $(m, m_0)$ or $(m, m)$ to a NE where neither insurer invests $(m_0, m_0)$. Mandatory information sharing can create situations of free-riding, because $(m, m)$ is not feasible (see region B).}
 \label{fig:comparison}
\end{figure}

These results illustrate important concerns about mandatory information sharing about cyber incidents in insurance. With no way to enforce data quality, a mandatory information sharing regime may be (i)~counterproductive if leading to no investment in data quality---an $(m_0, m_0)$ NE---or (ii)~unfair if leading to free-riding---an $(m, m_0)$ NE. This result is a bit reminiscent of \citet{laube2016economics}, who show that mandatory security breach reporting to authorities may have adverse effects unless it can be determined whether the failure to report an incident is deliberate concealment or mere lack of knowledge.

Needless to say, the model employed is simplified, and does not capture the full complexity of reality. \hlbreakable{A first aspect of this} is that it is exceedingly difficult to determine the parameters needed to facilitate exact calculations according to the model. Nevertheless, the model does show what we believe are important qualitative properties of mandatory information sharing situations: potentially undesirable transitions to counterproductive or unfair NEs if mandatory information sharing is enacted. \hlbreakable{A second aspect of this relates to the normal approximation discussed in Section~\ref{normal-discussion}. A third aspect is that the model treats duopoly rather than the more general oligopoly situation. Now, there is good reason to believe that our results should generalize to the oligopoly situation---\citet[see 3.a, p.~263]{raith1996general} has shown that for Cournot markets, the results from \citet{gal1986information} and similar duopoly  studies are valid for oligopolies as well. A detailed investigation of the oligopoly case, however, is beyond the scope of this paper. Nonetheless, 
 we would expect that insurers would have lower incentives to invest in data quality for two reasons: 1) the profit of firms will be lower with each additional competitor; and 2) 
imposing data sharing in an oligopoly will give each firm access to more data.
}

\hlbreakable{
In addition to the simplifications of the model, the NE concept also has some limitations. 
Concretely, the existence of a particular NE only says that if the players are there, none of them has anything to gain from unilaterally deviating. However, if they are not there, the NE concept does not provide any mechanism for how the NE could be reached. This means that it is not possible to say which outcome will actually occur when there are multiple NEs. 
Despite this, our analysis is important because it reveals strategic tensions between the players. 
For example, in practice, situations that create opportunities for free-riding may result in no investments, because each firm will try to free-ride.
}

Future research directions include analyzing whether sharing information policies benefit insurers and consumers (despite of creating free-riding scenarios) and designing incentives that could improve the data quality. Also, it would be interesting to contemplate other cost and risk aversion functions on the part of the insurers, \hlbreakable{as well as extending the treatment in Section~\ref{normal-discussion} of alternative cost distributions}.


\appendix

\section{Additional results and proofs}

The following results show some properties of the decision of each firm.
We use these results to analyze the equilibria of the market with different data sharing policies.

\begin{lemma}\label{lemma:properties_eq}
Let $m_i$ be a feasible decision in the game $\mathcal{G}_1$ when the adversary selects $m_j$. The following is satisfied
\begin{itemize}
\item $m_i=0$ is not feasible

\item  $m_i=m_0$ is feasible if $J_i ' (m_0, m_j) \geq 0$

 \item $0<m_i<m_0$ is feasible if 
$  J_i ' (m_i, m_j) = 0$ and
$  J_i '' (m_i, m_j) \leq 0$,
\end{itemize}
where $J_i'$ and $J_i''$ represent the first and second derivative of $J_i$ with respect to $m_i$.
\end{lemma}

\begin{proof}[Proof of \cref{lemma:properties_eq}]
We formulate the decision of each firm with the following optimization problem
\begin{equation}
 \begin{matrix}
  \max\limits_{m_i} & J_i(m_i, m_j) \\
  \text{s. t.} & -m_i \leq 0 \\
                    & m_i-m_0 \leq 0
 \end{matrix}
\end{equation}
whose Laplacian is
\begin{equation}
 L(m_i, \mu_1, \mu_2) = -J_i(m_i, m_j) - m_i \mu_1 + (m_i-m_0)\mu_2 .
\end{equation}
The noise level
$m_i$ is a local maxima if the following necessary conditions are satisfied
\begin{equation}
 -J_i ' (m_i, m_j) -\mu_1 + \mu_2 = 0
\end{equation}
and
\begin{equation}
- m_i \mu_1 + (m_i-m_0)\mu_2 = 0 
\end{equation}
with $\mu_1, \mu_2 \geq 0$.
Moreover, the second order sufficient condition is
\begin{equation}
 J_i''(m_i, m_j)  \leq 0.
\end{equation}

Now let us analyze the strategies that the insurer $i$ may choose:
\begin{itemize}
 \item $m_i=0$ is a valid solution if $\mu_2=0$, $\mu_1 \geq 0$ and 
\begin{equation}
  -J_i ' (0, m_j) -\mu_1 = 0
\end{equation}
however, this is not possible because
\begin{equation}
 \lim_{m_i\rightarrow 0} J_i '(m_i, m_j) \rightarrow \infty.
\end{equation}

\item $m_i=m_0$ is a valid solution if $\mu_2\geq0$, $\mu_1 = 0$ and 
\begin{equation}
  -J_i ' (m_0, m_j) +\mu_2 = 0
\end{equation}
which means that 
\begin{equation}
  J_i ' (m_0, m_j) \geq 0
\end{equation}

\item $0<m_i<m_0$ is a valid solution if $\mu_2=0$, $\mu_1 = 0$ and 
\begin{equation}
  J_i ' (m_i, m_j) = 0
\end{equation}
with the second order condition 
\begin{equation}
  J_i '' (m_i, m_j) \leq 0.
\end{equation}
\end{itemize}
\end{proof}

The next result shows that, except a special case, $(\hat m, m_0)$ and $(m_0, m_0)$ cannot be NE simultaneously. Thus, we assume that only one of these NE can occur. 
\begin{lemma}\label{eq:excluding_NE}
The tuples $(\hat m, m_0)$ and $(m_0, m_0)$, with $\hat m < m_0$, can be NE simultaneously only if 
they return the same profit, i.e.,  $J_i(\hat m, m_0)=J_i(m_0, m_0)$. Otherwise, only one of them can be a NE.
\end{lemma}
\begin{proof}[Proof of \cref{eq:excluding_NE}]
Assume that $(m_0, m_0)$ is a NE, this means that
\begin{equation}
J_i(m_0, m_0) \geq J_i(\hat m, m_0). 
\end{equation}
If the inequality is strict, i.e., $J_i(m_0, m_0) > J_i(\hat m, m_0)$, then $(\hat m, m_0)$ cannot be a NE. 
We can make the same argument when $(\hat m, m_0)$ is a NE.
\end{proof}

Now we introduce some results that specify conditions that restrict the possible NE.
\begin{corollary}\label{corol:special_eq}
If $J_i'(m, m)\geq 0$ for all $m\in[0, m_0]$, then $(m_0, m_0)$ and $(m_i, m_j)$, with $m_i, m_j<m_0$, are the only possible NE.
\end{corollary}
\begin{proof}[Proof of \cref{corol:special_eq}]
Recall that from \cref{remark:equilibria} the possible equilibria are $(m_0, m_0)$, $(\hat m, \hat m)$, $(\hat m, m_0)$, and $(m_i, m_j)$.
If $J_i'(m, m) \geq 0$ for all $m$, then $(m_0, m_0)$ is a feasible NE, but $(\hat m, \hat m)$, for some $\hat m \neq m_0$, is not feasible  (see \cref{lemma:properties_eq}). Moreover, from \cref{eq:excluding_NE} we know that an equilibria of the form $(\hat m, m_0)$ is not feasible. Thus, the only remaining possibilities are $(m_0, m_0)$ and $(m_i, m_j)$.
\end{proof}

\begin{lemma}\label{lemma:single_NE_extreme}
If $J_i'(m_i, m_j)$ is decreasing with respect to $m_j$, then the best response of the insurer $i$, denoted $m_i^*(m_j)$, is decreasing with respect to $m_j$. In such cases, if $(m_0, m_0)$ is a feasible NE, then no other NE exists.
\end{lemma}
\begin{proof}
Suppose that $m_i^*(m_j)<m_0$ is the best response to the strategy $m_j\in[0, m_0]$. This means that $m_i^*(m_j)$ is a local maxima, i.e., $J_i'(m_i^*(m_j), m_j)=0$ and $J_i''(m_i^*(m_j), m_j)<0$ (see \cref{lemma:properties_eq}). 
Then the following applies for two strategies $\tilde m_j$ and $\hat m_j$
\begin{equation}\label{eq:FOC_best_resp}
    J_i'(m_i^*(\hat m_j), \hat m_j) = 
    J_i'(m_i^*(\tilde m_j), \tilde m_j) = 0.
\end{equation}
Let us assume without loss of generality that $\tilde m_j < \hat m_j$.
If $J_i'(m_i, m_j)$ is decreasing wrt $m_j$, then 
\begin{equation}\label{eq:FOC_dec_m_j}
    J_i'(m_i^*(\tilde m_j), \tilde m_j) \geq  J_i'(m_i^*(\tilde m_j), \hat m_j).
\end{equation}
Now,
replacing \cref{eq:FOC_dec_m_j} in \cref{eq:FOC_best_resp} we obtain
\begin{equation}
    J_i'(m_i^*(\hat m_j), \hat m_j) \geq
    J_i'(m_i^*(\tilde m_j), \hat m_j).
\end{equation}
Since $J_i''(m_i^*(m_j), m_j)<0$ we know that $J_i'(m_i^*(\hat m_j), \hat m_j)$ is decreasing wrt $m_i^*(\hat m_j)$. Hence, 
\begin{equation}
m_i^*(\hat m_j) \leq m_i^*(\tilde m_j).
\end{equation}
In other words, the best response function $m_i^*(m_j)$ is decreasing wrt $m_j$.
\end{proof}

\subsection{Market with information sharing}

Now we are ready to prove \cref{lemma:eq_sharing}.
\begin{proof}[Proof of \cref{lemma:eq_sharing}]
In this proof we first find the form of the possible NE and then find the conditions to guarantee that they are feasible.
From \cref{eq:J_i_sharing} the first and second derivatives of the profit are:
\begin{equation}\label{eq:J_i_1st_sharing}
 \frac{ \partial }{ \partial m_i } J_i(m_i, m_j) =
 \frac{\sigma^2}{9b} \frac{ -m_j^2 }{ ( (\sigma+m_i)(\sigma+m_j)-\sigma^2 )^2 } +
 \frac{1}{m_i \log(\alpha)}
\end{equation}
and 
\begin{equation}
 \frac{ \partial^2 }{ \partial m_i^2 } J_i(m_i, m_j) =
 \frac{\sigma^2}{9b} \frac{  2 m_j^2 (\sigma + m_j) }{ ( (\sigma+m_i)(\sigma+m_j)-\sigma^2 )^3 } -
 \frac{1}{m_i^2 \log(\alpha)}.
\end{equation}

Let us show that a NE in which both firms invest ($m_i, m_j\neq m_0$) does not exist.
Note that $(m_i, m_j)$, with $0 < m_i, m_j < m_0$, is a feasible equilibria if it satisfies $ J_i ' (m_i, m_j) =  0$ for each firm. This FOC can be rewritten as (see \cref{eq:J_i_1st_sharing}) 
\begin{equation}\label{eq:foc_sharing}
 \frac{1}{m_i m_j^2 \log(\alpha)} = \frac{\sigma^2}{9b} \frac{ 1 }{ ( (\sigma+m_i)(\sigma+m_j)-\sigma^2 )^2 }.
\end{equation}
Since the right hand side is identical for each firm we obtain
\begin{equation}
 \frac{m_j \log(\alpha)}{m_i \log(\alpha)} 
 = \frac{ m_j^2 }{ m_i^2 } .
\end{equation}
The above expression implies that 
$m_i = m_j$ (this is the only possible solution).
Replacing $m_i=m$ in \cref{eq:foc_sharing} results in
\begin{equation}\label{eq:eq_case2}
 \frac{1}{m \log(\alpha)} = \frac{\sigma^2}{9b} \frac{ 1 }{ ( 2\sigma + m )^2 } .
\end{equation}

Now, the second derivative evaluated on $m_i=m$ is 
\begin{equation}
 \frac{ \partial^2 }{ \partial m_i^2 } J_i(m, m) =
 \frac{\sigma^2}{9b} \frac{  2 (\sigma + m) }{ m( 2\sigma + m )^3 } -
 \frac{1}{m^2 \log(\alpha)}.
\end{equation}
Replacing  \cref{eq:eq_case2} in the previous expression leads to
\begin{align}
 \frac{ \partial^2 }{ \partial m_i^2 } J_i(m, m) & =
 \frac{\sigma^2}{9b} \frac{  2 (\sigma + m) }{ m( 2\sigma + m )^3 } -
 \frac{1}{m}  \frac{\sigma^2}{9b} \frac{ 1 }{ ( 2\sigma + m )^2 } \\
 & =
 \frac{\sigma^2}{9b} \frac{  1 }{ ( 2\sigma + m )^3 } > 0.
\end{align}
Thus, a pair $(m, m)$ that satisfies the FOC corresponds to a local minimum; hence, it is not a valid NE.
This implies that in the NE at least one of the firms doesn't invest at all. Following this observation we set $m_j=m_0$ and investigate possible values of $m_i$ in the NE.
A pair $(m_i, m_0)$ is a feasible NE if it satisfies the FOC
\begin{equation}
 \frac{ \partial }{ \partial m_i } J_i(m_i, m_0) =
 \frac{1}{m_i \log(\alpha)} - 
 \frac{\sigma^2}{9b} \frac{ m_0^2 }{ ( \sigma m_0 + m_i(\sigma+m_0) )^2 } 
 = 0 .
\end{equation}
The above expression leads to
%
\begin{equation}
9 b ( \sigma m_0 + m_i(\sigma+m_0) )^2 - \sigma^2 m_0^2 m_i \log \alpha = 0,
\end{equation}
which is a 
quadratic equation of the form 
\begin{equation}\label{eq:quadratic}
 A m_i^2 + B m_i + C = 0,
\end{equation}
with $ A = 9b (\sigma + m_0)^2$, $ B = 18 b \sigma m_0 (\sigma+m_0) - \sigma^2 m_0^2 \log \alpha$, and $ C = 9 b \sigma^2 m_0^2$.
%
%
The solution of \cref{eq:quadratic} has the well known form
\begin{equation}\label{eq:m_i_sol}
 m_i = \frac{-B \pm \sqrt{B^2 - 4 A C}}{2 A},
\end{equation}
where
\begin{equation}
 B^2 - 4 A C = 
 \sigma^3 m_0^3 \log \alpha (\sigma m_0 \log \alpha - 36 b(\sigma + m_0)) .
\end{equation}

Now, let us investigate the conditions in which \cref{eq:m_i_sol} has no valid solution. In other words, cases in which $J_i'(m_i, m_0) > 0 $ for all $m_i\in[0, m_0]$. 
If this happens then the game can have a single NE, namely $(m_0, m_0)$ (see \cref{lemma:properties_eq}).
\cref{eq:m_i_sol} has an imaginary value if $B^2 - 4 A C < 0$, that is, if
\begin{equation}
 \sigma (m_0 \log \alpha - 36 b) - 36b m_0 < 0.
\end{equation}
This inequality holds in the following two cases:
\begin{itemize}
 \item If $m_0 \log \alpha - 36 b > 0$ and 
 \begin{equation}
  \sigma < \frac{36 b m_0}{ m_0 \log \alpha - 36 b } = \tilde \sigma.
 \end{equation}
 
 \item If $m_0 \log \alpha - 36 b < 0$.
\end{itemize}

In  summary, $(m_0, m_0)$ is the only NE if 
$ m_0 < \frac{36 b}{ \log \alpha }  $
or if $ m_0 > \frac{36 b}{ \log \alpha }  $
and $\sigma > \tilde \sigma$.

Now, 
a real solution exists only if $B^2 - 4 A C \geq 0$, that is, 
if
 \begin{equation}\label{eq:cond_real_sol_sharing}
  m_0 \geq \frac{36 b}{\log \alpha} \quad \text{and} \quad
  \sigma \geq \frac{36 b m_0}{ m_0 \log \alpha - 36 b }.
 \end{equation}
Moreover, observe that 
$|B| > \sqrt{B^2 - 4AC}$ (since $AC>0$). Therefore, if
$B<0$ then we have two positive solutions. This occurs if
\begin{equation}
 \sigma m_0 (18 b (\sigma + m_0) - \sigma m_0 \log \alpha) < 0.
\end{equation}
The above holds when 
\begin{equation}\label{eq:cond_B_neg_sharing}
m_0 >  \frac{18 b}{\log \alpha}
\quad \text{and} \quad
\sigma > \frac{18 b m_0}{ m_0 \log \alpha - 18 b }.
\end{equation}
Note that the conditions in \cref{eq:cond_B_neg_sharing} hold when 
\cref{eq:cond_real_sol_sharing} is true. For this reason, if \cref{eq:quadratic} has real solutions, then they are positive.

Now, let us find the conditions to have a feasible NE of the form $(m_i, m_0)$, where $m_i<m_0$. 
Alternatively, we are looking for situations in which $m_i=m_0$ is not a feasible solution, i.e., when $J_i'(m_0, m_0) \leq 0$.
This is the case when \cref{eq:quadratic} has only one solution in the interval $(0, m_0]$. In other words, if the largest root is greater than  $m_0$:
\begin{equation}
 \frac{ -B + \sqrt{B^2 - 4 A C} }{2 A} \geq m_0.
\end{equation}
From the above inequality we obtain
\begin{equation}
 B^2 - 4AC \geq (B+2 A m_0)^2,
\end{equation}
which leads to 
\begin{equation}\label{eq:extreme_eq_cond}
 0 \geq  A m_0^2 + B m_0 + C.
\end{equation}
We expand 
\cref{eq:extreme_eq_cond} to obtain
\begin{equation}
0 \geq (2 \sigma)^2 + 2 (2\sigma) m_0 + m_0^2 - \sigma^2 m_0 \frac{\log\alpha}{9 b}
\end{equation}
The above in turn leads to 
\begin{equation}\label{eq:cond_internal_sol_sharing}
 0 \geq (m_0 + \sigma (\Gamma+2) )( m_0 -  \sigma(\Gamma-2) )
\end{equation}
where $\Gamma^2 = \frac{m_0 \log \alpha}{9 b}$.
Observe that $\Gamma > 2$ when \cref{eq:cond_real_sol_sharing} has real roots. Therefore,  \cref{eq:cond_internal_sol_sharing} holds if
\begin{equation}\label{eq:cond_internal_sol_sharing_b}
 \sigma \geq \frac{m_0}{\Gamma-2} = \hat \sigma.
\end{equation}
Thus, from \cref{eq:cond_real_sol_sharing,eq:cond_internal_sol_sharing_b} $(m_i, m_0)$ is a feasible solution when 
\begin{equation}
   m_0 \geq \frac{36 b}{\log \alpha}
  \quad \text{and} \quad 
   \sigma \geq \frac{m_0}{\Gamma-2} = \hat \sigma
\end{equation}
%
%
%
%
%
\end{proof}

\subsection{Market without information sharing}

Let us define $x_i = m_i + \sigma$, $x_j = m_j + \sigma$, $\bar x = \sigma + m_0$,  and the functions
\begin{equation}
 f_1(x_i, x_j) = \frac{4 x_j x_i + \sigma^2}{ 4 x_j x_i - \sigma^2 }
\end{equation}
and 
\begin{equation}
 f_2 (x_i, x_j) = \frac{(2 x_j - \sigma)^2}{ (4 x_j x_i - \sigma^2) ^ 2 }.
\end{equation}
to write the marginal profit as
\begin{equation}\label{eq:partial_J_rn_nonsharing}
 \frac{ \partial }{ \partial m_i} J_i (m_i, m_j) = 
 \frac{1}{x_i - \sigma} \frac{1}{\log \alpha}
- \frac{ \sigma^2 }{ b } 
f_1(x_i, x_j) f_2(x_i, x_j) 
\end{equation}

Before proving \cref{lemma:eq_non_sharing} we need the following results:
\begin{lemma}\label{eq:sol_sigma}
The equation
\begin{equation}\label{eq:quadratic_sigma_a}
 (2 m_0 + 3\sigma)^2 = \sigma^2 \gamma
\end{equation}
has no positive solution (i.e., $\sigma\geq 0$) when $9-\gamma > 0$.
However, if $9-\gamma > 0$, then there is only one positive solution
$\sigma = \frac{2 m_0}{ \gamma^{1/2} - 3 }$.
\end{lemma}

\begin{proof}
We can rewrite \cref{eq:quadratic_sigma_a} as
\begin{equation}\label{eq:quadratic_sigma}
 4 m_0^2 + 12 m_0 \sigma + (9-\gamma) \sigma^2 = 0 
\end{equation}
Let $A = 9-\gamma$, $B=12 m_0$ and $C=4 m_0^2$. Since $m_0>0$, then $B, C > 0$.
The solution to \cref{eq:quadratic_sigma} has the form
\begin{equation}
 \sigma = \frac{-B \pm \sqrt{B^2 - 4 A C}}{ 2 A}.
\end{equation}
If $A>0$, then $\sqrt{B^2 - 4 A C}<B$; hence, both solution of  \cref{eq:quadratic_sigma} are negative.
On the other hand, if $A<0$, then $\sqrt{B^2 - 4 A C} > B$. In this case, \cref{eq:quadratic_sigma} has only one positive solution:
\begin{align}
 \sigma = & \frac{-B - \sqrt{B^2 - 4 A C}}{ 2 A} \\
 = & \frac{2 m_0}{ \gamma^{1/2} - 3 }.
\end{align}
\end{proof}

\begin{lemma}\label{lemma:feasibility_free_rider}
  A game $\mathcal{G}_1$ where insurers do not share information can have a NE of the form  $(m, \frac{\sigma^2}{4 m})$ if $\frac{\sqrt{3}-1}{2} \sigma \leq m \leq \frac{1+\sqrt{3}}{4} \sigma$.
\end{lemma}

\begin{proof}
From the second derivative of the profit $J_i''(m_i, m_j)$ and the FOC
$ J_i'(m_i, m_j) = 0$ we get
\begin{equation}
 \frac{\partial^2}{\partial m_i^2} J_i(m_i, m_j) = 
 \frac{\sigma^2}{b} \frac{ (2 x_j - \sigma)^2 }{ (4 x_i x_j - \sigma^2)^4 (x_i-\sigma) } g(x_i, x_j),
\end{equation}
where $g(x_i, x_j) =  16 x_j (\sigma^2 + 2 x_i x_j) (x_i - \sigma) - (4 x_i x_j + \sigma^2) (4 x_i x_j - \sigma^2)$. 
Observe that we only need to identify the sign of $g(x_i, x_j)$.

Now, 
replacing $m_j = \frac{\sigma^2}{4 m_i}$, i.e., $x_j = \frac{\sigma^2}{4 (x_i - \sigma)} + \sigma$,  and $x_i = m + \sigma$ we get
\begin{equation}
 g(m, \frac{\sigma^2}{4 m}) = 
 \frac{\sigma^2 (8 m^2 - 4 \sigma m - \sigma^2) ( 2 m^2 + 4 \sigma m + \sigma^2) }{ m^2 }
\end{equation}
If $m \leq \frac{1+\sqrt{3}}{4} \sigma$ then 
$8 m^2 - 4 \sigma m - \sigma^2 \leq 0$
and  $(m_i, \frac{\sigma^2}{4 m_i})$ is a local maximum.

Replacing $m_i = \frac{\sigma^2}{4 m_j}$, i.e., 
$x_i = \frac{\sigma^2}{4 (x_j - \sigma)} + \sigma$  and $x_j = m + \sigma$, we get
\begin{equation}
 g(\frac{\sigma^2}{4 m}, m) = 
 \frac{\sigma^2 (-2 m^2 - 2 \sigma m + \sigma^2) ( 8 m^2 + 8 \sigma m + \sigma^2) }{ m^2 }
\end{equation}
 If $m \geq \frac{\sqrt{3}-1}{2} \sigma$ then 
$-2 m^2 - 2 \sigma m + \sigma^2\leq0$
and  $(\frac{\sigma^2}{4 m}, m)$ is a local maxima.

In summary, $(m, \frac{\sigma^2}{4 m})$ is a feasible NE if $ \frac{\sqrt{3}-1}{2} \sigma \leq m \leq \frac{1+\sqrt{3}}{4} \sigma$.
\end{proof}

\begin{lemma}\label{lemma:feasibility_symmetric}
 A game $\mathcal{G}_1$ where insurers do not share information can have a NE of the form  $(m, m)$ if $\sigma>2 m$.
\end{lemma}
\begin{proof}
 From the second derivative of the profit $J_i''(m_i, m_j)$ and the FOC
$ J_i'(m_i, m_j) = 0$ we get
\begin{equation}
 J_i''(m, m) = \frac{\sigma^2 (2 x - \sigma)^2}{b (x-\sigma)(4 x^2 - \sigma^2)^4} \left\{ 16 x (x-\sigma)(2 x^2 + \sigma^2) - (4 x^2 + \sigma^2) (4 x^2 - \sigma^2) \right\},
\end{equation}
where $x=m+\sigma$. Expanding the right hand side
\begin{equation}
 16 x^4 -  32 x^3 \sigma + 16 \sigma^2 x^2 - 16 x \sigma^3 +  \sigma^4
\end{equation}
We can reorganize as
\begin{equation}
 (2x - \sigma)^4 - 8 x \sigma^2 (x + \sigma)
\end{equation}
Replacing $x=m+\sigma$ we obtain
\begin{equation}
 (2 m + \sigma )^4 - 8 \sigma^2 (m + \sigma)(m+2\sigma)
\end{equation}
Then, if $\sigma>2 m$ we can guarantee that 
$J_i''(m, m)\leq 0$, which means that a NE o the form $(m, m)$ is feasible.
\end{proof}

\begin{lemma}\label{lemma:eq_equivalence}
Consider a game $\mathcal{G}_1$ without sharing information. If $(m, \frac{\sigma^2}{4 m})$ is a feasible solution, then $(m, m)$ is also a feasible solution, but the converse is not necessarily true. 
\end{lemma}
\begin{proof}
 Let $(m, \frac{\sigma^2}{4 m})$ be a feasible solution, which means that
 \begin{equation}\label{eq:foc_a}
  J_i' \left( m, \frac{\sigma^2}{4 m} \right) = 0
 \end{equation}
 \begin{equation}\label{eq:foc_b}
  J_i' \left( \frac{\sigma^2}{4 m}, m \right) = 0
 \end{equation}

Without loss of generality we assume that $m\leq \frac{\sigma^2}{4 m}$.
 Note that 
the marginal profit is decreasing wrt $m_j$, that is, 
\begin{equation}
\frac{\partial}{\partial m_j} J_i'(m_i, m_j) = \frac{4 \sigma^3 (2 x_j - \sigma) (\sigma^3 -4\sigma^2x_i + 8 \sigma x_i x_j - 8 x_i^2 x_j) }{ b (4 x_i x_j - \sigma^2)^4 }<0
\end{equation}
The inequality follows since $x_i \geq \sigma$.

Now, since $m< \frac{\sigma^2}{4 m}$, from \cref{eq:foc_a,eq:foc_a} we get
\begin{equation}
0 =  J_i' \left( m, \frac{\sigma^2}{4 m} \right) \leq J_i'(m, m)
\end{equation}
and
\begin{equation}
 0 = J_i' \left( \frac{\sigma^2}{4 m}, m \right) \geq J_i' \left( \frac{\sigma^2}{4 m}, \frac{\sigma^2}{4 m} \right).
\end{equation}

Since the function is continuous wrt $m_i$ and $m_j$, then there exist a  $m$ such that
$J_i'(m, m) \geq J_i'(m, m)=0 \geq  J_i'(\frac{\sigma^2}{4 m}, \frac{\sigma^2}{4 m})$.
\end{proof}

Now we are ready to prove  \cref{lemma:eq_non_sharing}.
\begin{proof}[Proof of \cref{lemma:eq_non_sharing}]
First, let us consider a scenario in which a pair $(m_i, m_j)$ with $m_i, m_j \in (0, m_0)$ is a NE, which satisfies the FOC $ J_i'(m_i, m_j) = 0$ for both firms (see \cref{lemma:properties_eq}).
We can reorganize \cref{eq:partial_J_rn_nonsharing} for each player to obtain
\begin{equation}
 \frac{-1}{ x_i - \sigma } \frac{1}{\log \alpha} \frac{1}{ (2 x_j - \sigma)^2 }
 =
 \frac{-1}{ x_j - \sigma } \frac{1}{\log \alpha} \frac{1}{(2 x_i - \sigma)^2}.
\end{equation}
The above expression leads to 
%
\begin{equation}
 \frac{(2 x_i - \sigma)^2}{ x_i-\sigma } = \frac{(2 x_j - \sigma)^2}{ x_j-\sigma },
\end{equation}
which has two solutions, namely 
$m_i=m_j$ and $m_j = \frac{\sigma^2}{4 m_i}$. Thus, the NE can have the following form:
\begin{equation}
 (\hat m, \hat m) \quad \text{and} \quad \left(\tilde m, \frac{\sigma^2}{4 \tilde m} \right),
\end{equation}
for some $\hat m, \tilde m < m_0$.

An equilibria of the form 
$(\tilde m, \frac{\sigma^2}{4 \tilde m})$ has multiple restrictions.
Concretely, it
is a local maxima only if $\frac{\sqrt{3} - 1 }{2}\sigma \leq \tilde m \leq \frac{1+\sqrt{3}}{4}\sigma$ (see 
\cref{lemma:feasibility_free_rider}). In addition, it 
restricts the range of $\tilde m$ to $[\frac{\sigma^2}{4m_0}, m_0]$; 
hence, such solutions are feasible only if $m_0>\frac{\sigma^2}{4 m_0}$ (i.e.,  $2 m_0 > \sigma $).

\cref{lemma:eq_equivalence} shows that if $(\tilde m, \frac{\sigma^2}{4 \tilde m})$ is a feasible NE, then $(\hat m, \hat m)$ is also feasible (however, the converse is not necessarily true).
Note that a NE of the form $(m, m)$ is a local maxima if $\sigma \geq 2 m$ (see 
\cref{lemma:feasibility_symmetric}).
We focus on equilibria of the form 
 $(\hat m, \hat m)$
 because it has less restrictions.

Now let us analyze the conditions for the following equilibria: $(m_0, m_0)$, 
$(\hat m, m_0)$, and $(\hat m, \hat m)$, for $\hat m\in(0, m_0)$.

Let us show that $J_i'(m, m) > 0$ for all $m$, which according to \cref{corol:special_eq}, is enough to guarantee that
$(m_0, m_0)$ is the only NE.
Note that
\begin{equation}
 J_i'(m, m) = \frac{1}{(x-\sigma)\log \alpha} - \frac{\sigma^2}{b} f_1(x, x) f_2(x, x),
\end{equation}
with $x=m+\sigma$.
Since $f_1(x, x)\leq \frac{5}{3}$ and $f_2(x, x)\leq \frac{1}{9 \sigma^2}$, then
\begin{align}
 J_i'(m, m) & \geq \frac{1}{m_0 \log \alpha} - \frac{\sigma^2}{b}\frac{5}{3}\frac{1}{9 \sigma^2} \\
 & = \frac{1}{m_0 \log \alpha} - \frac{5}{27 b}
\end{align}
Then $ J_i'(m, m)\geq 0$ if $m_0\leq \frac{27b}{5 \log \alpha}$.
Note that if $m_0> \frac{27b}{5 \log \alpha}$, then an equilibria of the form $(\hat m, \hat m)$ is possible.

Now, using $f_2(\bar x, \bar x) \leq \frac{5}{3}$ we can 
obtain the following
\begin{equation}\label{eq:marginal_m0_lower}
 \frac{ \partial }{ \partial m_i} J_i (m_0, m_0) \geq
 \frac{1}{m_0} \frac{1}{\log \alpha}
- \frac{ \sigma^2 }{ b } 
\frac{1}{ (2 m_0 + 3 \sigma)^2 } \frac{5}{3}
\end{equation}
Let $\acute \sigma$ be such that the upper bound in \cref{eq:marginal_m0_upper} is equal to zero, that is, 
\begin{equation}\label{eq:tilde_sigma_a}
 \frac{1}{m_0} \frac{1}{\log \alpha}
- \frac{ \acute \sigma^2 }{ b } 
\frac{1}{ (2 m_0 + 3 \acute \sigma)^2 } \frac{5}{3} = 0
\end{equation}
From \cref{eq:sol_sigma} we know that the previous expression has a positive solution if $9 - \tilde \gamma < 0 $, 
where $\tilde \gamma = \frac{5 m_0 \log \alpha}{3 b}$.
Concretely, the solution to \cref{eq:tilde_sigma_a} is
\begin{equation}
 \acute \sigma = \frac{2 m_0}{ \tilde \gamma^{1/2} - 3 }
\end{equation}
if 
\begin{equation}
 m_0 \geq \frac{27 b}{5 \log \alpha}.
\end{equation}
In summary, 
\begin{equation}
\frac{ \partial }{ \partial m_i} J_i (m_0, m_0) \geq 0 
\end{equation}
for $\sigma \leq \acute \sigma$, and according to \cref{lemma:properties_eq}, $(m_0, m_0)$ is a feasible NE.
Moreover, since $J_i'$ is decreasing wrt $m_j$, then $(m_0, m_0)$ is the only NE (see \cref{lemma:single_NE_extreme}).


Next, using $1 < f_2(\bar x, \bar x)$ we obtain
\begin{equation}\label{eq:marginal_m0_upper}
 \frac{ \partial }{ \partial m_i} J_i (m_0, m_0) <
 \frac{1}{m_0} \frac{1}{\log \alpha}
- \frac{ \sigma^2 }{ b } 
\frac{1}{ (2 m_0 + 3 \sigma)^2 }
\end{equation}
Let $\breve \sigma$ such that the upper bound in \cref{eq:marginal_m0_upper} is equal to zero
\begin{equation}
 \frac{1}{m_0} \frac{1}{\log \alpha}
- \frac{ \breve \sigma^2 }{ b } 
\frac{1}{ (2 m_0 + 3 \breve \sigma)^2 } = 0
\end{equation}
We can rewrite the above as 
\begin{equation}\label{eq:hat_sigma_a}
 (2 m_0 + 3 \breve \sigma)^2 = \breve \sigma^2 \hat \gamma,
\end{equation}
where $\hat \gamma = \frac{m_0 \log \alpha}{b}$.
From \cref{eq:sol_sigma} we know that the previous expression has a positive solution if $9 - \hat \gamma < 0 $, that is, if
\begin{equation}
 m_0 > \frac{9 b }{ \log \alpha}.
\end{equation}
If the above is satisfied, then the solution is 
\begin{equation}
 \breve \sigma = \frac{2 m_0}{ \hat \gamma^{1/2} - 3 }
\end{equation}
Note that \cref{eq:marginal_m0_upper} is decreasing wrt $\sigma$; hence, 
\begin{equation}
\frac{ \partial }{ \partial m_i} J_i (m_0, m_0) < 0 
\end{equation}
for $\sigma \geq \breve \sigma$, in which case $(\hat m, m_0)$ is a feasible NE.
%
%
%
%
\end{proof}

\modification{
Lastly, the next result defines conditions to have an $\varepsilon$ Nash Equilibrium:
\begin{lemma}\label{lemma:eNE}Consider the games 
$\mathcal{G}_a = \langle \mathcal{P}, (S_i)_{i\in\mathcal{P}}, (u_i)_{i\in\mathcal{P}} \rangle$
and 
$\mathcal{G}_b = \langle \mathcal{P}, (S_i)_{i\in\mathcal{P}}, (\tilde u_i)_{i\in\mathcal{P}} \rangle$.
If $\tilde u_i$ is an approximation of $u_i$ s.t.
\begin{align}\label{eq:bounds}
u_i(s_i, s_{-i}) \geq \tilde u_i(s_i, s_{-i}) - \delta \\
u_i(s_i, s_{-i}) \leq \tilde u_i(s_i, s_{-i}) + \eta
\end{align}
for $\delta, \eta>0$ and any $s_i\in S_i$ for $i\in\mathcal{P}$,
then a NE of $\mathcal{G}_b$ is an $\varepsilon$-NE (or near NE) of $\mathcal{G}_a$.
\end{lemma}
\begin{proof}[Proof of \cref{lemma:eNE}]
Let $(\tilde s_i, \tilde s_{-i})$ be a NE for $\mathcal{G}_b$. This means that 
\begin{equation}
\tilde u_i(\tilde s_i, \tilde s_{-i}) \geq \tilde u_i(s_i, \tilde s_{-i}), \quad \forall i\in\mathcal{P}, s_i\in S_i
\end{equation}
From the upper and lower bound conditions in \cref{eq:bounds} we get
\begin{equation}
u_i(\tilde s_i, \tilde s_{-i}) + \delta \geq u_i(s_i, \tilde s_{-i}) -\eta
\end{equation}
which can be rewritten as 
\begin{equation}
u_i(\tilde s_i, \tilde s_{-i}) \geq u_i(s_i, \tilde s_{-i}) -\eta - \delta.
\end{equation}
Then $(\tilde s_i, \tilde s_{-i})$ is an $\varepsilon$-NE, with $\varepsilon=\eta+\delta$. 
\end{proof}
}

\section*{Acknowledgments}

This research was supported by L{\"a}nsf{\"o}rs{\"akringar} (O.~Reinert \& T.~Wiesinger), the Swedish Foundation for Strategic Research, grant no. SM19-0009 (U.~Franke), and Digital Futures (U.~Franke \& C.~Barreto).
\hlbreakable{We would like to thank the anonymous reviewers for their insightful feedback, which helped us to improve the manuscript.}


\DeclareRobustCommand\EIOPAlongname{ European Insurance and Occupational Pensions Authority}

\bibliography{References}

\end{document}